\documentclass[12pt,reqno]{article}

\usepackage[usenames]{color}
\usepackage{graphicx}

\usepackage[colorlinks=true,
linkcolor=webgreen,
filecolor=webbrown,
citecolor=webgreen]{hyperref}

\definecolor{webgreen}{rgb}{0,.5,0}
\definecolor{webbrown}{rgb}{.6,0,0}

\usepackage{color}
\usepackage{fullpage}
\usepackage{float}
\usepackage{enumerate}

\usepackage{amsmath,amssymb}
\usepackage{amsthm}
\usepackage{amsfonts}
\usepackage{latexsym}
\usepackage{epsf}
\usepackage{fullpage}

\newcommand{\pal}{\textsc{Pal}}
\newcommand{\maxpal}{\textsc{MaxPal}}
\newcommand{\factoreq}{\textsc{FactorEq}}
\newcommand{\occurs}{\textsc{Occurs}}
\newcommand{\border}{\textsc{Border}}
\newcommand{\predin}{\textsc{In}}
\newcommand{\subs}{\textsc{Subs}}
\newcommand{\rich}{\textsc{Rich}}
\newcommand{\priv}{\textsc{Priv}}
\newcommand{\privtwo}{\textsc{Priv}'}
\newcommand{\closed}{\textsc{Closed}}

\newcommand{\rtsp}{\textsc{RtSp}}
\newcommand{\minr}{\textsc{MinRt}}

\newcommand{\unrepsuf}{\textsc{UnrepSuf}}
\newcommand{\minunrepsuf}{\textsc{MinUnrepSuf}}
\newcommand{\trap}{\textsc{Trap}}
\newcommand{\ucf}{\textsc{UCF}}
\newcommand{\unbal}{\textsc{Unbal}}
\newcommand{\uniqpref}{\textsc{UniquePref}}
\newcommand{\uniqsuff}{\textsc{UniqueSuff}}

\DeclareMathOperator{\Th}{Th}

\def\Enn{\mathbb{N}}

\author{Luke Schaeffer \\
Massachusetts Institute of Technology\\
Cambridge, MA 02139\\
USA\\
\href{mailto:lrschaeffer@gmail.com}{\tt lrschaeffer@gmail.com}
\and Jeffrey Shallit\\
School of Computer Science\\
University of Waterloo\\
Waterloo, ON N2L 3G1 \\
Canada\\
\href{mailto:shallit@cs.uwaterloo.ca}{\tt shallit@cs.uwaterloo.ca}
}

\title{Closed, Palindromic, Rich, Privileged, Trapezoidal, and Balanced Words in Automatic Sequences}

\begin{document}

\maketitle

\theoremstyle{plain}
\newtheorem{theorem}{Theorem}
\newtheorem{corollary}[theorem]{Corollary}
\newtheorem{lemma}[theorem]{Lemma}
\newtheorem{proposition}[theorem]{Proposition}

\theoremstyle{definition}
\newtheorem{definition}[theorem]{Definition}
\newtheorem{example}[theorem]{Example}
\newtheorem{conjecture}[theorem]{Conjecture}

\theoremstyle{remark}
\newtheorem{remark}[theorem]{Remark}

\begin{abstract}
We prove that the property of being closed (resp., palindromic, rich, privileged
trapezoidal, balanced) is expressible in 
first-order logic for automatic (and some related) sequences.
It therefore follows that the characteristic function of those $n$ for
which an automatic sequence $\bf x$ has a closed (resp., palindromic,
privileged, rich,
trapezoidal, balanced) factor of length $n$ is automatic.  
For privileged words this requires a new characterization of the
privileged property.  We compute the corresponding characteristic functions for
various famous sequences, such as the Thue-Morse sequence, the
Rudin-Shapiro sequence, the ordinary paperfolding sequence, the
period-doubling sequence, and the Fibonacci sequence.  Finally, we also
show that the function counting the total number of palindromic factors
in a prefix of length $n$ of a $k$-automatic sequence is not $k$-synchronized.
\end{abstract}

\section{Introduction}

Recently a wide variety of different kinds of words have been
studied in the combinatorics on words literature,
including the six flavors of the title:  closed, palindromic, rich, privileged,
trapezoidal, and balanced words.   In this paper we show that, for
$k$-automatic sequences $\bf x$ (and some analogs, such as the
so-called ``Fibonacci-automatic'' sequences \cite{Du&Mousavi&Schaeffer&Shallit:2014}), the property of a
factor belonging to each class is expressible in first-order logic; more
precisely, in the theory $\Th(\Enn, +, n \rightarrow {\bf x}[n])$.
Previously we did this for unbordered factors \cite{Goc&Mousavi&Shallit:2013}.

As a consequence, we get that (for example) the characteristic sequence
of those lengths for which a factor of that length belongs to each class
is $k$-automatic, and the number of such factors of each length forms
a $k$-regular sequence.  (For definitions, see, for example,
\cite{Allouche&Shallit:2003}.)

Using an implementation of a decision procedure for first-order expressible
properties, we can give explicit expressions for the lengths of factors
in each class for some famous sequences, such as the Thue-Morse sequence,
the Rudin-Shapiro sequence, the period-doubling sequence, and the 
ordinary paperfolding sequence.  For some of the properties, these
expressions are surprisingly complicated.

\section{Notation and definitions}
\label{defs}

As usual, if $w = xyz$, we say that $x$ is a prefix of $w$,
that $z$ is a suffix of $w$, and $y$ is a factor of $w$.
By $|x|_w$ we mean the number of
(possibly overlapping) occurrences of $w$ as a factor of $x$.  For example,
$|{\tt confrontation}|_{\tt on} = 3$.
By $x^R$ we mean the reversal (sometimes called mirror image) of the 
word $x$.  Thus, for example, $({\tt drawer})^R = {\tt reward}$.
By $\Sigma_k$ we mean the alphabet $\{0,1,\ldots, k-1\}$ of
cardinality $k$.

A factor $w$ of $x$ is said to be {\it right-special} if both
$wa$ and $wb$ are factors of $x$, for two distinct letters $a$ and $b$.

A word $x$ is a {\it palindrome\/} if $x = x^R$.
Examples of palindromes in English include {\tt radar} and {\tt redivider}.
Droubay, Justin, and Pirillo 
\cite{Droubay&Justin&Pirillo:2001}
proved that every word of length $n$ contains at most $n+1$ distinct
palindromic factors (including the empty word).
A word is called {\it rich} if it contains exactly this many.  
For example, the  English words {\tt logology} and {\tt Mississippi} are
both rich.
For example, {\tt Mississippi} has the
following distinct nonempty palindromic factors:
$${\tt M,\ i,\ s,\ p,\ ss,\ pp,\ sis,\ issi,\ ippi,\ ssiss,\ ississi}. $$
For more about rich words, see 
\cite{Glen&Justin&Widmer&Zamboni:2009,deLuca&Glen&Zamboni:2008,Bucci&DeLuca&Glen&Zamboni:2009,Bucci&deLuca&DeLuca:2009}.

A nonempty word $w$ is a {\it border} of a word $x$ if 
$w$ is both a prefix and a suffix of $x$.
A word $x$ is called {\it closed} (aka ``complete
first return'') if it is of length $\leq 1$,
or if it has a border $w$ with $|x|_w = 2$.
For example, {\tt abracadabra} is closed because of the border {\tt abra},
while {\tt alfalfa} is closed because of the border {\tt alfa}. 
The latter example shows that, in the definition, the prefix and suffix
are allowed to overlap.  For more about closed words,
see \cite{Badkobeh&Fici&Liptak:2015}.

A word $x$ is called {\it privileged} if it is of length $\leq 1$,
or it has a border $w$ with $|x|_w = 2$ that is itself privileged.
Clearly every privileged word is closed, but {\tt mama} is an 
example of an English word that is
closed but not privileged.  For more about privileged words,
see \cite{Kellendonk&Lenz&Savinien:2013,Peltomaki:2013,Peltomaki:2015,Forsyth&Jayakumar&Peltomaki&Shallit:2015}.

A word $x$ is called {\it trapezoidal} if it has,
for each $n \geq 0$, at most $n+1$ distinct factors of length $n$.
Since for $n=1$ the definition requires at most $2$ distinct factors,
this means that every trapezoidal word can be defined over an alphabet
of at most $2$ letters.  An example in English is the word {\tt deeded}.  
See, for example, \cite{deLuca:1999,dAlessandro:2002,deLuca&Glen&Zamboni:2008,Bucci&DeLuca&Fici:2013}.

A word $x$ is called {\it balanced} if, for all factors $y,z$ of the
same length of $x$ and all letters $a$ of the alphabet,  the inequality
$\left| |y|_a - |z|_a \right| \leq 1$ holds.  Otherwise it is {\it unbalanced}.
An example of a balanced word in English is {\tt banana}.

We use the terms ``infinite sequence'' and ``infinite word'' as synonyms.
In this paper, names of infinite words are given in the {\bf bold\/} font.
All infinite words are
indexed starting at position $0$.  If ${\bf x} = x_0 x_1 x_2 \cdots$ is
an infinite word, with each $x_i$ a single letter, then
by ${\bf x}[i..j]$ for $j \geq i-1$ we mean the finite word
$x_i x_{i+1} \cdots x_j$.  By $[i..j]$ we mean the set $\{i, i+1, \ldots, j\}$.

\section{Sequences}

In this section we define the five sequences we will study.  For
more information about these sequences,
see, for example, \cite{Allouche&Shallit:2003}.

The {\it Thue-Morse sequence}
${\bf t} = t_0 t_1 t_2 \cdots = {\tt 01101001} \cdots$
is defined by
the relations $t_0 = 0$, $t_{2n} = t_n$, and $t_{2n+1} = 1-t_n$.  It
is also expressible as the fixed point, starting with $0$, of the
morphism $\mu:0 \rightarrow 01$, $1 \rightarrow 10$.

The {\it Rudin-Shapiro sequence}
${\bf r} = r_0 r_1 r_2 \cdots = {\tt 00010010} \cdots$ is
defined by the relations $r_0 = 0$, 
$r_{2n} = r_n$, $r_{4n+1} = r_n$, $r_{8n+7} = r_{2n+1}$,
$r_{16n+3} = r_{8n+3}$, $r_{16n+11} = r_{4n+3}$.  
It is also expressible as the image, under the coding 
$\tau:  n \rightarrow \lfloor n/2 \rfloor$, of the fixed point,
starting with $0$, of the morphism
$\rho:0 \rightarrow 01$, $1 \rightarrow 02$, $2 \rightarrow 31$, 
$3 \rightarrow 32$.

The {\it ordinary paperfolding sequence}
${\bf p} = p_0 p_1 p_2 \cdots = {\tt 00100110} \cdots$
is defined by the relations $p_0 = 0$, $p_{2n+1} = p_n$,
$p_{4n} = 0$, $p_{4n+2} = 1$.
It is also expressible as the image, under the coding
$\tau$ above, of the fixed point, starting with $0$, of the morphism
$\rho: 0 \rightarrow 01$, $1 \rightarrow 21$, $2 \rightarrow 03$, 
$3 \rightarrow 23$.

The {\it period-doubling sequence}
${\bf d} = d_0 d_1 d_2 \cdots  = {\tt 10111010} \cdots$
is defined by the relations $d_0 = 1$,
$d_{2n} = 1$, $d_{4n+1} = 0$, and $d_{4n+3} = d_n$.
It is also expressible as the fixed point, starting with $1$,
of the morphism $\delta:1 \rightarrow 10$, $0 \rightarrow 11$.

The {\it Fibonacci sequence}
${\bf f} = f_0 f_1 f_2 \cdots = {\tt 01001010} \cdots$
is the fixed point, starting with $0$, of the morphism
$\varphi:0 \rightarrow 01$, $1 \rightarrow 0$.

\section{Common predicates}

Before we see how rich words, privileged words, closed words, etc. can
be phrased as first-order predicates, let us define a few basic predicates.

First, we have the two basic predicates $\predin(i,r,s)$, which is true iff
$i \in [r..s]$:
$$\predin(i,r,s) := (i \geq r) \ \wedge\ (i\leq s),$$
and
$\subs(i,j,m,n)$, which is true iff $[i..i+m-1] \subseteq [j..j+n-1]$:
$$
\subs(i,j,m,n) := (j \leq i) \ \wedge\ (i+m \leq j+n).
$$

Next, we have the predicate
$$\factoreq(i,j,n) := \forall k \ (k<n) \implies ({\bf x}[i+k] = {\bf x}[j+k]),$$
which checks whether ${\bf x}[i..i+n-1]$ and ${\bf x}[j..j+n-1]$ are equal by comparing them at corresponding positions, ${\bf x}[i+k]$ and ${\bf x}[j+k]$, for $k = 0, \ldots, n-1$. By a similar principle, we can compare ${\bf x}[i..i+n-1]$ with ${\bf x}[j..j+n-1]^{R}$, but in this paper we only need the special case $i = j$, i.e., palindromes:
$$
\pal(i,n) := \forall k \ (k<n) \implies ({\bf x}[i+k] = {\bf x}[i+n-1-k]).
$$

From $\factoreq$, we derive other useful predicates. For instance, the predicate
$$
\occurs(i,j,m,n) := (m \leq n) \ \wedge \ 
(\exists k \ (k+m \leq n) \ \wedge \ \factoreq(i,j+k,m))
$$
tests whether ${\bf x}[i..i+m-1]$ is a factor of ${\bf x}[j..j+n-1]$. We also define
\begin{align*}
\border(i,m,n) &:= \predin(m,1,n) \ \wedge \ \factoreq(i,i+n-m,m),
\end{align*}
which is true iff ${\bf x}[i..i+m-1]$ is a border of ${\bf x}[i..i+n-1]$.

In the next five sections, we obtain our results using the
implementation of a decision procedure for the corresponding properties,
written by Hamoon Mousavi, and called {\tt Walnut}, to prove theorems
by machine computation.  The software is available for download at

\centerline{\url{https://cs.uwaterloo.ca/~shallit/papers.html} \ .}

All of the predicates in this paper can easily be translated into
Hamoon Mousavi's {\tt Walnut} program.  Files for the examples in
this paper are available at the same URL as above, so the reader can
easily run and verify the results.

\section{Closed words}

We can create a predicate $\closed(i,n)$ that asserts that
${\bf x}[i..i+n-1]$ is closed as follows:

$$
(n\leq 1) \ \vee \ (\exists j \ (j<n) \ \wedge\ \border(i,j,n) \ \wedge \neg \occurs(i,i+1,j,n-2)) \\
$$

\begin{theorem}
\begin{enumerate}[(a)]
\item There is a closed factor of Thue-Morse of every length.

\item There is a 15-state automaton accepting the base-$2$ representation
of those $n$ for which there is a closed factor of Rudin-Shapiro of length
$n$.

\item There is an 11-state automaton accepting the base-$2$ representation
of those $n$ for which there is a closed factor of the paperfolding sequence
of length $n$.  It is depicted below in Figure~\ref{Pclosed}.

\item There is a closed factor of the period-doubling sequence of 
every length.

\item There is a closed factor of the Fibonacci sequence of
every length.
\end{enumerate}
\end{theorem}

\begin{figure}[H]
\leavevmode
\def\epsfsize#1#2{0.6#1}
\centerline{\epsfbox{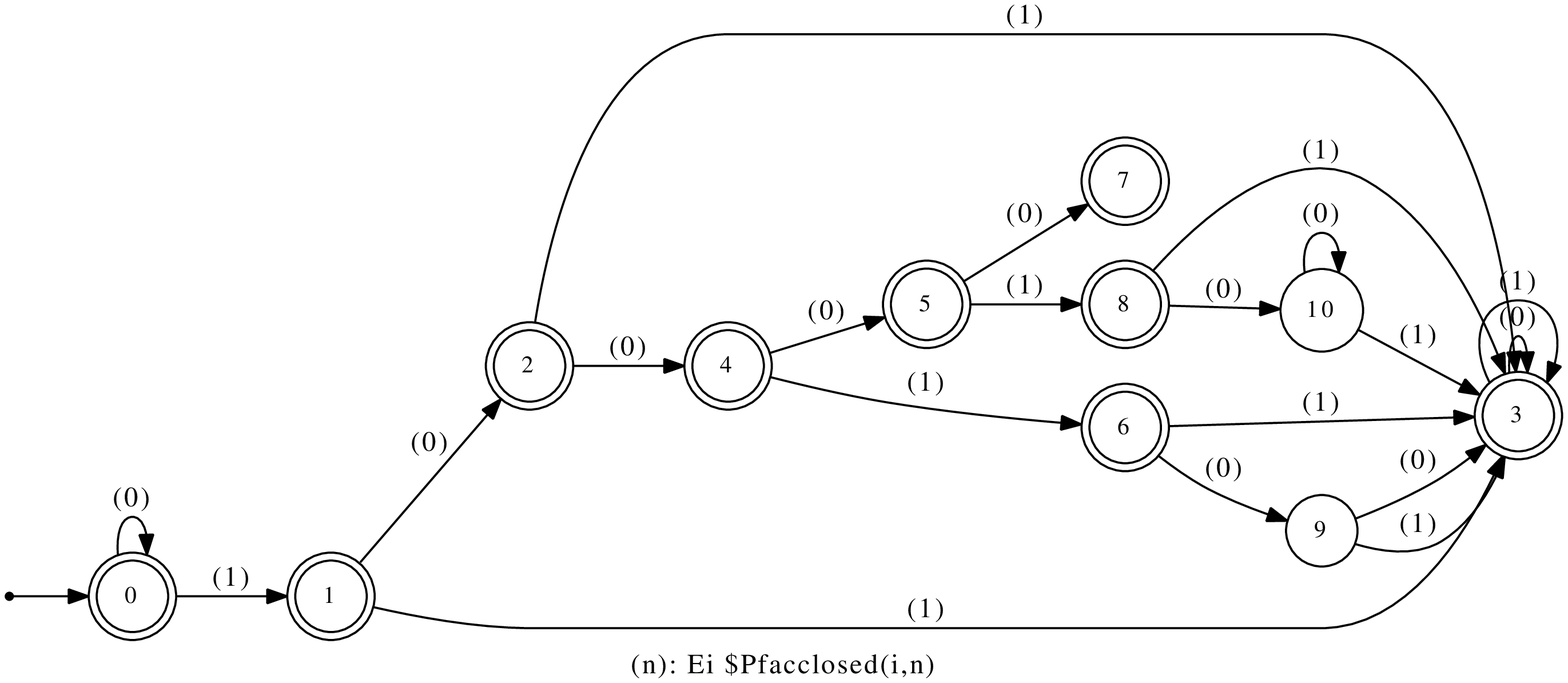}}
\caption{Automaton for lengths of closed factors of the paperfolding sequence}
\protect\label{Pclosed}
\end{figure}

As we have seen above, the Thue-Morse sequence contains a closed
factor of every length.  We now turn to enumerating $f(n)$, the number of
such factors of length $n$.  Here are the first few values of $f(n)$:

\begin{center}
\begin{tabular}{c|cccccccccccccccc}
$n$ & 0 & 1 & 2 & 3 & 4 & 5 & 6 & 7 & 8 & 9 & 10 & 11 & 12 & 13 & 14 & 15 \\
\hline
$f(n)$ & 1 & 2 & 2 & 2 & 4 & 4 & 6 & 4 & 8 & 8 & 10 & 8 & 12 & 8 & 8 &  8 
\end{tabular}
\end{center}

The first step is to create a predicate $\ucf(i,n)$ which is true
if ${\bf t}[i..i+n-1]$ is a closed factor of $\bf t$ of length $n$,
and is also the first occurrence of that factor:
$$ \ucf(i,n) := \closed(i,n) \ \wedge \  \neg \occurs(i,0,n,i+n-1).$$

The associated DFA then gives us (as in \cite{Goc&Mousavi&Shallit:2013})
a linear
representation for $f(n)$:  vectors $v, w$ and a matrix-valued
homomorphism $\mu:\lbrace 0, 1 \rbrace \rightarrow \Enn^{k \times k}$
such that $f(n) = v \mu(x) w^T$ for all $x$ that are valid base-$2$
representations of $n$.

They are as follows (with $\mu(i) = M_i$):

\begin{displaymath}
M_0 = \left[
{\scriptsize
\begin{array}{ccccccccccccccccccccccccccccccc}
1&1&0&0&0&0&0&0&0&0&0&0&0&0&0&0&0&0&0&0&0&0&0&0&0&0&0&0&0&0&0\\
0&0&0&0&1&0&0&0&0&0&0&0&0&0&0&0&0&0&0&0&0&0&0&0&0&0&0&0&0&0&0\\
0&0&0&0&0&0&0&1&1&0&0&0&0&0&0&0&0&0&0&0&0&0&0&0&0&0&0&0&0&0&0\\
0&0&0&0&0&0&0&0&0&0&1&1&0&0&0&0&0&0&0&0&0&0&0&0&0&0&0&0&0&0&0\\
0&0&0&0&0&0&0&0&0&0&0&0&0&0&0&0&0&0&0&0&0&0&0&0&0&0&0&0&0&0&0\\
0&0&0&0&0&0&0&0&0&0&0&0&0&0&0&0&1&1&0&0&0&0&0&0&0&0&0&0&0&0&0\\
0&0&0&0&0&0&0&0&0&0&0&0&0&0&0&0&0&0&0&0&1&0&0&0&0&0&0&0&0&0&0\\
0&0&0&0&0&0&0&0&0&0&0&0&0&0&0&0&0&0&0&0&1&1&0&0&0&0&0&0&0&0&0\\
0&0&0&0&0&0&0&0&0&0&0&0&0&0&0&0&0&0&0&0&0&0&1&1&0&0&0&0&0&0&0\\
0&0&0&0&0&0&0&0&0&0&0&0&0&0&0&0&0&0&0&0&0&0&0&0&1&0&0&0&0&0&0\\
0&0&0&0&0&0&0&0&0&0&0&0&0&0&0&0&0&0&0&0&0&0&0&0&1&0&0&1&0&0&0\\
0&0&0&0&0&0&0&0&0&0&0&0&0&0&0&0&0&0&0&0&0&0&0&0&1&0&0&0&0&1&0\\
0&0&0&0&0&0&0&0&0&0&0&0&1&0&0&0&0&0&0&0&0&0&0&0&0&0&0&0&0&0&1\\
0&0&0&0&0&0&0&0&0&0&0&0&0&0&0&0&0&0&0&0&0&0&1&0&0&0&0&0&0&0&0\\
0&0&0&0&0&0&0&0&0&0&0&0&0&0&0&0&0&0&0&0&0&0&0&0&0&0&0&0&0&0&0\\
0&0&0&0&0&0&0&0&0&0&0&0&0&0&0&0&0&0&0&0&0&0&0&0&0&0&0&0&0&1&0\\
0&0&0&0&0&0&0&0&0&0&0&0&0&0&0&0&0&0&0&0&1&0&0&1&0&0&0&0&0&0&0\\
0&0&0&0&0&0&0&0&0&0&0&0&0&0&0&0&0&0&0&0&0&0&1&0&0&0&0&0&1&0&0\\
0&0&0&0&0&0&0&0&0&0&0&0&0&0&0&0&0&0&0&0&0&0&0&0&1&0&0&0&0&0&0\\
0&0&0&0&0&0&0&0&0&0&0&0&0&0&0&0&0&0&0&0&0&0&1&0&0&1&0&0&0&0&0\\
0&0&0&0&0&0&0&0&0&0&0&0&0&0&0&0&0&0&0&0&0&0&0&0&0&1&0&0&0&0&0\\
0&0&0&0&0&0&0&0&0&0&0&0&1&0&0&0&0&0&0&0&0&0&0&0&0&0&0&0&1&0&0\\
0&0&0&0&0&0&0&0&0&0&0&0&0&0&0&0&0&0&0&0&0&0&1&0&0&0&0&0&0&0&0\\
0&0&0&0&0&0&0&0&0&0&0&0&1&0&0&0&0&0&0&0&0&0&0&0&0&0&0&0&1&0&0\\
0&0&0&0&0&0&0&0&0&0&0&0&0&0&0&0&0&0&0&0&0&0&1&0&0&0&0&0&0&0&0\\
0&0&0&0&0&0&0&0&0&0&0&0&0&0&0&0&0&0&0&0&0&0&0&0&0&1&0&0&0&0&0\\
0&0&0&0&0&0&0&0&0&0&0&0&0&0&0&0&0&0&0&0&0&0&0&0&0&2&0&0&0&0&0\\
0&0&0&0&0&0&0&0&0&0&0&0&0&0&0&0&0&0&0&0&0&0&0&0&0&0&0&0&0&0&0\\
0&0&0&0&0&0&0&0&0&0&0&0&0&0&0&0&0&0&0&0&0&0&0&0&0&0&0&0&1&0&1\\
0&0&0&0&0&0&0&0&0&0&0&0&0&0&0&0&0&0&0&0&0&0&0&0&0&0&0&0&0&0&0\\
0&0&0&0&0&0&0&0&0&0&0&0&0&0&0&0&0&0&0&0&0&0&0&0&0&0&0&0&0&0&2\\
\end{array}
}
\right]
\end{displaymath}

\begin{displaymath}
M_1 = \left[
{\scriptsize
\begin{array}{ccccccccccccccccccccccccccccccc}
0&0&1&1&0&0&0&0&0&0&0&0&0&0&0&0&0&0&0&0&0&0&0&0&0&0&0&0&0&0&0\\
0&0&0&0&0&1&1&0&0&0&0&0&0&0&0&0&0&0&0&0&0&0&0&0&0&0&0&0&0&0&0\\
0&0&0&0&0&0&0&0&0&1&0&0&0&0&0&0&0&0&0&0&0&0&0&0&0&0&0&0&0&0&0\\
0&0&0&0&0&0&0&0&0&0&0&0&1&1&0&0&0&0&0&0&0&0&0&0&0&0&0&0&0&0&0\\
0&0&0&0&0&0&0&0&0&0&0&0&0&0&1&1&0&0&0&0&0&0&0&0&0&0&0&0&0&0&0\\
0&0&0&0&0&0&0&0&0&0&0&0&0&0&0&0&0&0&1&1&0&0&0&0&0&0&0&0&0&0&0\\
0&0&0&0&0&0&0&0&0&0&0&0&0&0&0&0&0&0&0&1&0&0&0&0&0&0&0&0&0&0&0\\
0&0&0&0&0&0&0&0&0&0&0&0&0&0&0&0&0&0&0&1&0&0&1&0&0&0&0&0&0&0&0\\
0&0&0&0&0&0&0&0&0&0&0&0&0&0&0&0&0&0&0&0&0&0&1&0&0&0&0&0&0&0&0\\
0&0&0&0&0&0&0&0&0&0&0&0&0&0&0&0&0&0&0&0&0&0&0&0&0&1&1&0&0&0&0\\
0&0&0&0&0&0&0&0&0&0&0&0&0&0&0&0&0&0&0&0&0&0&0&0&0&0&1&0&1&0&0\\
0&0&0&0&0&0&0&0&0&0&0&0&0&0&0&0&0&0&0&0&0&0&0&1&0&0&0&0&0&0&0\\
0&0&0&0&0&0&0&0&0&0&0&0&1&0&0&0&0&0&0&0&0&0&1&0&0&0&0&0&0&0&0\\
0&0&0&0&0&0&0&0&0&0&0&0&0&0&0&0&0&0&0&0&0&0&0&0&0&1&1&0&0&0&0\\
0&0&0&0&0&0&0&0&0&0&0&0&0&0&0&0&0&0&0&0&0&0&0&0&0&0&0&0&0&1&0\\
0&0&0&0&0&0&0&0&0&0&0&0&0&0&0&0&0&0&0&0&0&0&0&1&0&0&0&0&0&1&0\\
0&0&0&0&0&0&0&0&0&0&0&0&0&0&0&0&0&0&0&1&0&0&1&0&0&0&0&0&0&0&0\\
0&0&0&0&0&0&0&0&0&0&0&0&1&0&0&0&0&0&0&0&0&0&0&0&0&0&0&0&0&0&0\\
0&0&0&0&0&0&0&0&0&0&0&0&0&0&0&0&0&0&0&0&0&0&0&0&0&0&0&0&0&0&0\\
0&0&0&0&0&0&0&0&0&0&0&0&0&0&0&0&0&0&0&0&0&0&0&1&0&0&0&0&0&0&0\\
0&0&0&0&0&0&0&0&0&0&0&0&0&0&0&0&0&0&0&0&0&0&0&1&0&0&0&0&0&0&0\\
0&0&0&0&0&0&0&0&0&0&0&0&1&0&0&0&0&0&0&0&0&0&1&0&0&0&0&0&0&0&0\\
0&0&0&0&0&0&0&0&0&0&0&0&0&0&0&0&0&0&0&0&0&0&0&0&0&0&0&0&0&0&0\\
0&0&0&0&0&0&0&0&0&0&0&0&1&0&0&0&0&0&0&0&0&0&1&0&0&0&0&0&0&0&0\\
0&0&0&0&0&0&0&0&0&0&0&0&0&0&0&0&0&0&0&0&0&0&0&0&0&0&0&0&0&0&0\\
0&0&0&0&0&0&0&0&0&0&0&0&0&0&0&0&0&0&0&0&0&0&0&0&0&1&0&0&1&0&0\\
0&0&0&0&0&0&0&0&0&0&0&0&0&0&0&0&0&0&0&0&0&0&0&0&0&0&0&0&2&0&0\\
0&0&0&0&0&0&0&0&0&0&0&0&0&0&0&0&0&0&0&0&0&0&0&0&0&1&1&0&0&0&0\\
0&0&0&0&0&0&0&0&0&0&0&0&0&0&0&0&0&0&0&0&0&0&0&0&0&0&0&0&0&0&2\\
0&0&0&0&0&0&0&0&0&0&0&0&0&0&0&0&0&0&0&0&0&0&0&0&0&1&1&0&0&0&0\\
0&0&0&0&0&0&0&0&0&0&0&0&0&0&0&0&0&0&0&0&0&0&0&0&0&0&0&0&0&0&2\\
\end{array}
}
\right]
\end{displaymath}

\begin{displaymath}
v = \left[
{\scriptsize
\begin{array}{ccccccccccccccccccccccccccccccc}
1&1&0&0&1&0&0&0&0&0&0&0&0&0&0&0&0&0&0&0&0&0&0&0&0&0&0&0&0&0&0
\end{array}
}
\right]
\end{displaymath}

\begin{displaymath}
w = \left[
{\scriptsize
\begin{array}{ccccccccccccccccccccccccccccccc}
1&0&1&1&0&0&0&0&1&0&0&0&1&1&0&0&0&1&0&0&0&1&1&0&0&0&0&1&0&0&1
\end{array}
}
\right]
\end{displaymath}

This linear representation can be minimized, using the algorithm in
\cite{Berstel&Reutenauer:2010}, obtaining

\begin{displaymath}
M'_0 = \left[
{\scriptsize
\begin{array}{ccccccccccccccccccccccccccccccc}
1&0&0&0&0&0&0&0&0&0 \\
0&0&1&0&0&0&0&0&0&0 \\
0&0&0&0&1&0&0&0&0&0 \\
0&0&0&0&0&0&1&0&0&0 \\
0&0&0&0&0&0&0&0&1&0 \\
0&0&0&0&0&-1&1&1&1/2&0 \\
0&0&0&0&0&-2&2&0&-3&4 \\
0&0&0&0&0&0&0&2&4&-4 \\
0&0&0&0&0&0&0&0&2&0 \\
0&0&0&0&0&0&0&1/2&11/4&-1
\end{array}
}
\right]
\end{displaymath}

\begin{displaymath}
M'_1 = \left[
{\scriptsize
\begin{array}{ccccccccccccccccccccccccccccccc}
0&1&0&0&0&0&0&0&0&0 \\
0&0&0&1&0&0&0&0&0&0 \\
0&0&0&0&0&1&0&0&0&0 \\
0&0&0&0&0&0&0&1&0&0 \\
0&0&0&0&0&0&0&0&0&1 \\
0&0&0&0&0&2&-2&-1&4&-2 \\
0&0&0&0&0&0&0&0&1&0 \\
0&0&0&0&0&4&-4&0&10&-8 \\
0&0&0&0&0&0&0&0&2&0 \\
0&0&0&0&0&1&-1&-1/2&7/2&-1
\end{array}
}
\right]
\end{displaymath}

\begin{displaymath}
v' = \left[
{\scriptsize
\begin{array}{ccccccccccccccccccccccccccccccc}
1&0&0&0&0&0&0&0&0&0
\end{array}
}
\right]
\end{displaymath}

\begin{displaymath}
w' = \left[
{\scriptsize
\begin{array}{ccccccccccccccccccccccccccccccc}
1&2&2&2&4&4&6&4&8&8
\end{array}
}
\right]
\end{displaymath}

From this, using technique in  \cite{Goc&Mousavi&Shallit:2013}, we can obtain the following relations 
\begin{eqnarray*}
f(8n) &=& -2f(2n+1) + f(4n) + 2f(4n+1) \\
f(8n+1) &=& -2f(2n+1) + 3 f(4n+1) \\
f(8n+3) &=& -2f(2n+1) +2f(4n+1) + f(4n+3) \\
f(8n+4) &=& 2f(2n+1) - \frac{5}{2}f(4n+1) + f(4n+2) + \frac{1}{2} f(4n+3) + f(8n+2) \\
f(8n+5) &=& 2f(4n+3) \\
f(8n+7) &=& -4f(2n+1) +2f(4n+1) - 2f(4n+3) +2f(8n+6) \\
f(16n+2) &=& -6f(2n+1) + \frac{13}{2} f(4n+1) + \frac{1}{2} f(4n+3) \\
f(16n+6) &=& -\frac{1}{2} f(4n+1) + f(4n+2) + \frac{3}{2} f(4n+3) + f(8n+2) \\
f(16n+10) &=& 2f(4n+3) + f(8n+6) \\
f(32n+14) &=& -2f(2n+1) -\frac{7}{2} f(4n+1) + 3 f(4n+2) +\frac{7}{2} f(4n+3)
	+ 3 f(8n+2) \\
f(32n+30) &=& 24f(2n+1) - 6f(4n+1) +14f(4n+3) -4f(8n+2) -12f(8n+6)+5f(16n+14) . 
\end{eqnarray*}

From these we can verify the following theorem by a
tedious induction on $n$:

\begin{theorem}
Let $n \geq 8$ and let $k \geq -1$ be an integer. Then
$$ f(n) = \begin{cases}
	2^{k+4}, & \text{if $15\cdot 2^k < n \leq 18 \cdot 2^k$}; \\
	2n - 20 \cdot 2^{k} - 2, &\text{if $18\cdot 2^k <
		n \leq 19 \cdot 2^k$}; \\
	56 \cdot 2^k - 2n + 2, & \text{if $19 \cdot 2^k <
		n \leq 20 \cdot 2^k$}; \\
	4n - 64\cdot 2^k - 4, & \text{if $20 \cdot 2^k < 
		n \leq 22 \cdot 2^k$}; \\
	112 \cdot 2^k - 4n + 4, & \text{if $22 \cdot 2^k < n \leq 24 \cdot 2^k$}; \\
	2^{k+4}, & \text{if $24\cdot 2^k < n \leq 28 \cdot 2^k$}; \\
	8n - 208 \cdot 2^k - 8, & \text{if $28 \cdot 2^k < n \leq 30 \cdot 2^k$}.
	\end{cases}
$$
\end{theorem}	

\section{Palindromic words}

Palindromes in words have a long history of being studied; for
example, see \cite{Allouche&Baake&Cassaigne&Damanik:2003}.

It is already known that many aspects of palindromes in $k$-automatic
sequences are expressible in first-order logic; see, for
example, \cite{Charlier&Rampersad&Shallit:2012}.  

In this section, we turn to a variation on palindromic words, the
so-called ``maximal palindromes''.  For us, a factor $x$ of an infinite
word $\bf w$ is a {\it maximal palindrome} if $x$ is a palindrome, while
no factor of the form $axa$ for $a$ a single letter occurs in $\bf w$.
This differs slightly from the existing definitions, which deal with
the maximality of {\it occurrences} \cite{I&Inenaga&Bannai&Takeda:2010}.

The property of being a maximal palindrome is easily expressible in terms
of predicates defined above:
$$\maxpal(i,n) := \pal(i,n) \ \wedge \ 
	(\forall j \ ((j \geq 1) \wedge \factoreq(i,j,n)) \implies
		{\bf x}[j-1] \not= {\bf x}[j+n]  )$$
Using this, and our program, we can easily prove the following result:

\begin{theorem}
\begin{itemize}
\item[(a)]  The Thue-Morse sequence contains maximal palindromes of
length $3 \cdot 4^n$ for each $n \geq 0$, and no others.  These
palindromes are of the form $\mu^{2n}(010)$ and $\mu^{2n}(101)$ for $n \geq 0$.

\item[(b)]  The Rudin-Shapiro sequence contains exactly
8 maximal palindromes.  They are
$$ 0100010,0001000,1110111,1011101,0010000100,1101111011,1110110111,10000100100001 .$$

\item[(c)]  The ordinary paperfolding sequence contains exactly
6 maximal palindromes.  They are
$$ 001100,110011,011000110,100111001,1000110110001,0111001001110 .$$

\item[(d)]  The period-doubling sequence contains maximal palindromes
of lengths $3 \cdot 2^n - 1$ for all $n \geq 0$, and no others.

\item[(e)]  The Fibonacci sequence contains no maximal palindromes
at all.
\end{itemize}
\end{theorem}

We now turn to a result about counting palindromes in automatic
sequences.  To state it, we first need to describe representations of
integers in base $k$.  By $(n)_k$ we mean the string over the alphabet
$\Sigma_k := \{ 0,1, \ldots, k-1 \rbrace$
representing $n$ in base $k$, and having
no leading zeroes.  This is generalized to representing
$r$-tuples of integers
by changing the alphabet to $\Sigma_k^r$, and padding shorter representations
on the left, if necessary, with leading zeroes.  Thus, for
example, $(6,3)_2 = [1,0][1,1][0,1]$.  By $[w]_k$, for a word $w$,
we mean the value of $w$ when interpreted as an integer in base $k$.

Next, we need the concept of
$k$-synchronization 
\cite{Carpi&Maggi:2001,Carpi&DAlonzo:2009,Carpi&Dalonzo:2010,Goc&Schaeffer&Shallit:2013}.  We say a function $f(n)$ is 
{\it $k$-synchronized}
if there is a finite automaton accepting the language
$\{ (n, f(n))_k \ : \ n \geq 0 \}$.

The following is a useful lemma:

\begin{lemma}
If $(f(n))_{n \geq 0}$ is a $k$-synchronized sequence, and
$f \not= O(1)$, then there exists a constant $c > 0$ such that
$f(n) \geq cn$ infinitely often.
\label{unsynch}
\end{lemma}

\begin{proof}
Since $f \not= O(1)$, there exists $n>0$ such that $f(n) > k^N$, where $N$
is the number of states in the minimal automaton
 accepting $L^R$, where $L = \{ (n, f(n))_k : n \geq 0 \}$.
 Apply the pumping lemma to the string  $z = (n, f(n))_k^R$.  
 It says that we can
 write $z = uvw$, where $|uv| \leq n$ and  $w$ has nonzero elements in both components.
 Then, letting $(n_i, f(n_i)) = [(uv^i w)^R]_k$ we see that this subsequence
 has the desired property. 
\end{proof}

\begin{theorem}
The function counting the number of distinct palindromes in a prefix
of length $n$ is not $k$-synchronized.
\end{theorem}

\begin{proof}
Our proof is based on two infinite words, ${\bf a} = (a_i)_{i \geq 0}$ and
${\bf b} = (b_i)_{i \geq 0}$.

The word ${\bf a}$ is defined as follows:
$$ a_i = \begin{cases}
	(k \bmod 2)+1, & \text{if there exists $k$ such that
		$4^{k+1} - 4^k \leq i \leq 4^{k+1} + 4^k$}; \\
		0, & \text{otherwise}.
	\end{cases}
$$

The word ${\bf b}$ is defined as follows:
$$ b_i = \begin{cases}
	(k \bmod 2)+1, & \text{if there exists $k$ such that
		$4^{k+1} - 4^k < i < 4^{k+1} + 4^k$}; \\
		0, & \text{otherwise}.
	\end{cases}
$$

We leave the easy proof that $\bf a$ and $\bf b$ are $4$-automatic
to the reader.

We now compare the palindromes in $\bf a$ to those in $\bf b$.
From the definition, every palindrome in either sequence is clearly in
$$ 0^* + 1^* + 2^* + 0^* 1^* 0^* + 0^* 2^* 0^* .$$

Since $\bf a$ has longer blocks of $1$'s and $2$s than $\bf b$ does,
there may be some palindromes of the form $1^i$ or $2^i$ that occur in
a prefix of $\bf a$,  but not the corresponding prefix of $\bf b$.
Conversely, $\bf b$ may contain palindromes of the form $0^i$ that do
not occur in the corresponding prefix of $\bf a$.

Call an occurrence of a factor in a word
{\it novel} if it is the first occurrence in the word.
The remaining palindromes (of the form $0^i 1^j 0^i$ or $0^i 2^j 0^i$)
must be centered at a position that is a
power of $4$.
It is not hard to see that if ${\bf a}[i..i+n-1]$ is a novel palindrome 
occurrence of this form in $\bf a$, then $\bf b[i..i+n-1]$ is also a novel 
palindrome occurrence of this form.

On the other hand, for each $k \geq 1$,
there are two palindromes that occur in $\bf b$ but not $\bf a$.
The first is of the form $0 1^j 0$ or $0 2^j 0$,
since the corresponding factor of $\bf a$ is either $1 \cdots 1$ or 
$2 \cdots 2$, and hence has been previously accounted for
Second, there is a factor of the form $0^*1^*0^*$ or $0^*2^*0^*$ which appears
as $20^*1^*0^*$ or $10^*2^*0^*$ in $\bf a$,
since the neighbouring block of $1$'s or $2$'s is slightly wider
and therefore slightly closer. We conclude that the length-$n$ prefix
of $\bf b$ has $2 \log_4 n + O(1)$
more palindromes than the length-$n$ prefix of $\bf a$.

Now suppose, contrary to what we want to prove, that
the number of palindromes in the prefix of length $n$ of a
$k$-automatic sequence is $k$-synchronized.
In particular, the sequence
$\bf a$ (resp., $\bf b$) is 4-automatic, so the number of palindromes in
${\bf a}[0..n-1]$ (resp., ${\bf b}[0..n-1]$ is
4-synchronized.
Now, using a result of
Carpi and Maggi \cite[Prop.~2.1]{Carpi&Maggi:2001},
the number of palindromes in
${\bf b}[1..n]$ minus the number of palindromes in ${\bf a}[1..n]$
is 4-synchronized.
But from above this difference is $2 \log_4 n + O(1)$, which by 
Lemma~\ref{unsynch} cannot be $4$-synchronized.
This is a contradiction.
\end{proof}

\section{Rich words}

As we have seen above, a word $x$ is rich iff it has $|x|+1$ distinct
palindromic subwords.  As stated, it does not seem easy to phrase
this in first-order logic.  Luckily, there is an alternative
characterization of rich words, which can be found in \cite[Prop.~3]{Droubay&Justin&Pirillo:2001}:
a word is rich if every prefix $p$ of $w$ has a palindromic suffix $s$
that occurs only once in $p$.  This property can be stated as follows:
\begin{align*}
\rich(i,n) &:= \forall m \ \predin(m,1,n) \implies \\
& \qquad (\exists j \ \subs(j,i,1,m) \ \wedge \pal(j,i+m-j) \ \wedge 
\ \neg \occurs(j,i,i+m-j,m-1)) .
\end{align*}

Finally, we can express the property that ${\bf x}$ has a rich factor
of length $n$ as follows:
$$ \exists i \ \rich(i,n) .$$

\begin{theorem}
\begin{enumerate}[(a)]
\item The Thue-Morse sequence contains exactly
161 distinct rich factors, the longest being of length $16$.

\item The Rudin-Shapiro sequence contains exactly
975 distinct rich factors, the longest being of length $30$.

\item The ordinary paperfolding sequence contains 
exactly 494 distinct rich factors, the longest being of length $23$.

\item The period-doubling sequence has a rich factor of every
length.  In fact, every factor of the period-doubling sequence is rich.

\item Every factor of the Fibonacci sequence is rich.
\end{enumerate}
\end{theorem}

Of course, (e) was already well known.

\section{Privileged words}

The recursive definition for privileged words given above in
Section~\ref{defs} is not obviously expressible in first-order logic.
However, we can prove a new, alternative characterization of these words,
as follows:

Let's say a word $w$ has property P if for all $n$,
$1 \leq n \leq |w|$, there exists a word $x$ such that 
$1 \leq |x| \leq n$, and $x$ occurs exactly once in the first $n$
symbols of $w$, as a prefix, and
$x$ also occurs exactly once in the last $n$ symbols of $w$, as a suffix.

\begin{lemma}
If $w$ is a bordered word with property $P$, then every
border also has property $P$. 
\end{lemma}

\begin{proof}
Let $z$ be a border of $w$. Given any $1 \leq n \leq |z|$, property $P$ for $w$ says that 
there exists a border $x$ of $w$ such that $1 \leq |x| \leq n$, and $x$ occurs exactly once 
in the first (resp., last) $n$ symbols in $w$. Then observe that the first 
(resp., last) $n$ 
symbols of $w$ are precisely the first (resp., last) $n$ symbols of $z$. Since $x$ is also 
a border of $z$, it follows that $z$ has property $P$.
\end{proof}

\begin{theorem}
A word $w$ is privileged if and only if it has property P.
\end{theorem}

\begin{proof}
If $w$ is privileged, then, by definition, there is a sequence of
privileged words $w = w_0, w_1, ..., w_{k-1}, w_k$ such that $|w_k| = 1$
and for all $i$, $w_{i+1}$ is a prefix and suffix of $w_i$ and
occurs nowhere else in $w_i$. Given an integer $n$, let 
$x$ be the largest $w_i$ such that $|w_i| \leq n$. Either $i = 0$ because 
$n = |w|$ and everything works out, or $|w_{i-1}| > n$. Then $w_i$ is
a prefix of $w_{i-1}$ (and therefore a prefix of $w$),
and there is no other occurrence of $w_i$ in $w_{i-1}$
(which includes the first $n$ symbols of $w$). Similarly, $w_i$
is a suffix of $w$, but does not occur again in the last $n$ symbols of $w$. 

For the other direction,
we assume the word has property P and use
induction on the length of $w$. If $|w| = 1$
then the word is privileged immediately. Otherwise, take $n = |w| - 1$
and find the corresponding $x$ promised by property P. Then $x$
is both a prefix and a suffix of $w$, so it has property P.
It is also shorter than $w$, so by induction, $x$ is privileged.
Then $x$ is a privileged prefix and suffix of $w$
which does not occur anywhere else in $w$ (by property P), so $w$
is privileged. 
\end{proof}

This property can be represented as a predicate in two different ways.
First, 
let's write a predicate that is true iff the prefix
${\bf x}[i..i+m-1]$ occurs exactly once in 
${\bf x}[i..i+n-1]$:
$$\uniqpref(i,m,n) := \forall j \ \predin(j,1,n-m-1)
\implies \neg \factoreq(i,i+j,m) .$$

There is a similar expression for whether the suffix ${\bf x}[i+n-m..i+n-1]$
occurs exactly once in ${\bf x}[i..i+n-1]$:
$$\uniqsuff(i,m,n) := \forall j \ \predin(j,1,n-m-1) \implies \neg\factoreq(i+n-m,i+n-m-j,m) .$$

And finally, our first characterization of privileged words is
\begin{multline*}
\priv(i,n) := (n \leq 1) \ \vee \ (\forall m \ \predin(m,1,n) \implies
\\
(\exists p \ \predin(p,1,m) \ \wedge \ \border(i,p,n) \ \wedge\ \uniqpref(i,p,m)
\ \wedge\ \uniqsuff(i+n-m,p,m) )).
\end{multline*}

Alternatively, we can write
\begin{multline*}
\privtwo(i,n) := (n \leq 1) \ \vee \ (\forall m \ \predin(m,1,n) \implies \\
(\exists p\ \predin(p,1,m) \ \wedge\ \border(i,p,n) \ \wedge 
\ \neg\occurs(i,i+1,p,m-1) \ \wedge 
\  \neg\occurs(i,i+n-m,p,m-1))) .
\end{multline*}

\begin{theorem}
\begin{enumerate}[(a)]
\item There is a $46$-state automaton accepting the base-$2$ expansions
of those $n$ for which the Thue-Morse sequence has a privileged factor
of length $n$.

\item There is an $84$-state automaton accepting the base-$2$ expansions
of those $n$ for which the Rudin-Shapiro sequence has a privileged factor
of length $n$.

\item There is a $47$-state automaton accepting the base-$2$ expansions
of those $n$ for which the paperfolding sequence has a privileged factor
of length $n$.

\item The set of $n$ for which the period-doubling sequence has
a privileged factor of length $n$ is
$$\{ 0, 2\} \ \cup \ \{2n+1 \ : \ n \geq 0 \}.$$
There is a $4$-state automaton accepting the base-$2$ expansions
of those $n$ for which the period-doubling sequence has a privileged factor
of length $n$.  It is illustrated below in Figure~\ref{pdpriv}.  

\item There is a 20-state automaton accepting the Zeckendorf
representations of those pairs $(i,n)$ for which ${\bf f}[i..i+n-1]$ is
privileged.  It is illustrated below in Figure~\ref{fib2priv}.
The Fibonacci word has privileged factors of every length.
If $n$ is even there is exactly one privileged factor.  If $n$ is
odd there are exactly two privileged factors.

\end{enumerate}
\end{theorem}

\begin{remark}  For (a)--(d) we used $\priv$ and for (e) we used
$\privtwo$.
\end{remark}

\begin{figure}[H]
\leavevmode
\def\epsfsize#1#2{1.0#1}
\centerline{\epsfbox{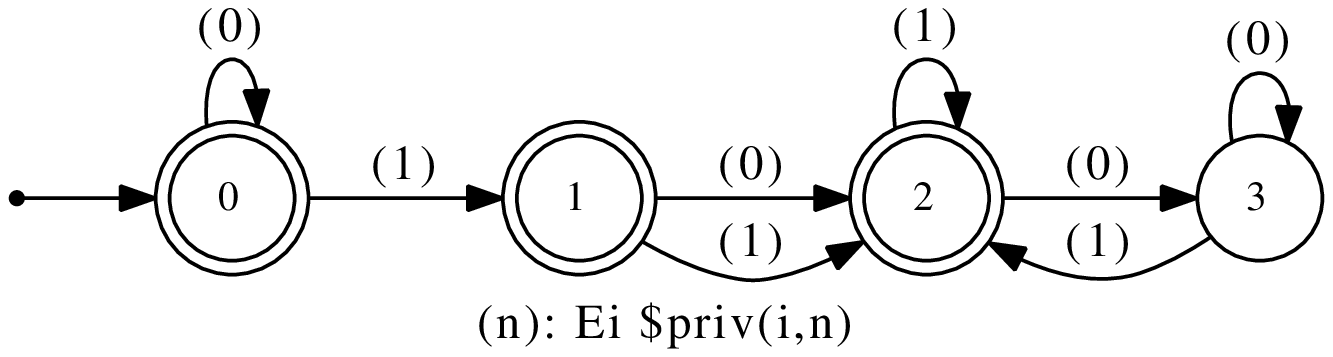}}
\caption{Automaton for lengths of privileged factors of the period-doubling
word}
\protect\label{pdpriv}
\end{figure}

\begin{figure}[H]
\leavevmode
\def\epsfsize#1#2{0.3#1}
\centerline{\epsfbox{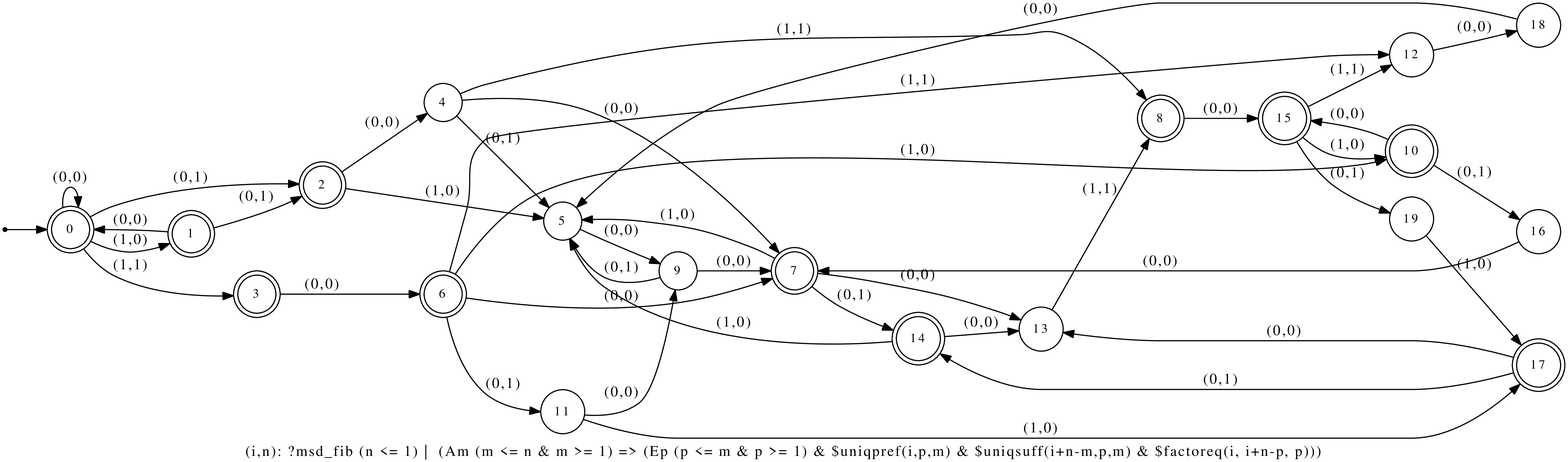}}
\caption{Automaton for privileged factors of the Fibonacci
word}
\protect\label{fib2priv}
\end{figure}

We now turn to recovering some of the results of \cite{Peltomaki:2015}
on the number $a(n)$ of privileged factors of the Thue-Morse sequence.
Here are the first few values of this sequence
\begin{center}
\begin{tabular}{c|ccccccccccccccccc}
$n$ & 0 & 1 & 2 & 3 & 4 & 5 & 6 & 7 & 8 & 9 & 10 & 11 & 12 & 13 & 14 & 15 & 16 \\
\hline
$a(n)$ & 1 & 2 & 2 & 2 & 2 & 0 & 4 & 0 & 8 & 0 & 8 & 0 & 4 & 0 & 0 & 0 &  0 
\end{tabular}
\end{center}

As we did above for closed words, we first make an automaton
for the first occurrences of each privileged factor of length
$n$.  We then convert this to a linear representation
$(v, \mu, w)$, obtaining

$$M_0 = 
\left[ 
\scriptsize{\begin{array}{cccccccccccccccccccccccccccccc}
1&1&0&0&0&0&0&0&0&0&0&0&0&0&0&0&0&0&0&0&0&0&0&0&0&0&0&0&0&0\\
0&0&0&0&1&0&0&0&0&0&0&0&0&0&0&0&0&0&0&0&0&0&0&0&0&0&0&0&0&0\\
0&0&0&0&0&0&0&1&1&0&0&0&0&0&0&0&0&0&0&0&0&0&0&0&0&0&0&0&0&0\\
0&0&0&0&0&0&0&0&0&0&1&1&0&0&0&0&0&0&0&0&0&0&0&0&0&0&0&0&0&0\\
0&0&0&0&0&0&0&0&0&0&0&0&0&0&0&0&0&0&0&0&0&0&0&0&0&0&0&0&0&0\\
0&0&0&0&0&0&0&0&0&0&0&0&0&0&0&0&1&1&0&0&0&0&0&0&0&0&0&0&0&0\\
0&0&0&0&0&0&0&0&0&0&0&0&0&0&0&0&0&0&0&0&1&0&0&0&0&0&0&0&0&0\\
0&0&0&0&0&0&0&0&0&0&0&0&0&0&0&0&0&0&0&0&1&1&0&0&0&0&0&0&0&0\\
0&0&0&0&0&0&0&0&0&0&0&0&0&0&0&0&0&0&0&0&0&0&1&1&0&0&0&0&0&0\\
0&0&0&0&0&0&0&0&0&0&0&0&0&0&0&0&0&0&0&0&0&0&0&0&1&0&0&0&0&0\\
0&0&0&0&0&0&0&0&0&0&0&0&0&0&0&0&0&0&0&0&0&0&0&0&1&1&0&0&0&0\\
0&0&0&0&1&0&0&0&0&0&0&0&0&0&0&0&0&0&0&0&0&0&0&0&1&0&0&0&0&0\\
0&0&0&0&0&0&0&1&1&0&0&0&0&0&0&0&0&0&0&0&0&0&0&0&0&0&0&0&0&0\\
0&0&0&0&0&0&0&0&0&0&0&0&0&0&0&0&0&0&0&0&0&0&1&0&0&0&0&0&0&0\\
0&0&0&0&0&0&0&0&0&0&0&0&0&0&0&0&0&0&0&0&0&0&0&0&0&0&0&0&0&0\\
0&0&0&0&1&0&0&0&0&0&0&0&0&0&0&0&0&0&0&0&0&0&0&0&0&0&0&0&0&0\\
0&0&0&0&0&0&0&0&0&0&0&0&0&0&0&0&0&0&0&0&1&0&0&1&0&0&0&0&0&0\\
0&0&0&0&0&1&0&0&0&0&0&0&0&0&0&0&0&0&0&0&0&0&1&0&0&0&0&0&0&0\\
0&0&0&0&0&0&0&0&0&0&0&0&0&0&0&0&0&0&0&0&0&0&0&0&1&0&0&0&0&0\\
0&0&0&0&0&0&1&0&0&0&0&0&0&0&0&0&0&0&0&0&0&0&0&0&0&0&0&0&1&0\\
0&0&0&0&0&0&1&0&0&0&0&0&0&0&0&0&0&0&0&0&0&0&0&0&0&0&0&0&0&0\\
0&0&0&0&0&0&0&0&1&0&0&0&0&0&0&0&1&0&0&0&0&0&0&0&0&0&0&0&0&0\\
0&0&0&0&0&0&0&0&0&0&0&0&0&0&0&0&0&0&0&0&0&0&0&0&0&0&0&0&0&1\\
0&0&0&0&0&0&0&0&1&0&0&0&0&0&0&0&1&0&0&0&0&0&0&0&0&0&0&0&0&0\\
0&0&0&0&0&0&0&0&0&0&0&0&0&0&0&0&0&0&0&0&0&0&0&0&0&0&0&0&1&0\\
0&0&0&0&0&0&0&0&0&0&0&0&0&0&0&0&0&0&0&0&0&0&0&0&0&0&0&0&0&0\\
0&0&0&0&0&0&0&0&1&0&0&0&0&0&0&0&1&0&0&0&0&0&0&0&0&0&0&0&0&0\\
0&0&0&0&0&0&0&0&0&0&0&0&0&0&0&0&0&0&0&0&0&0&0&0&1&1&0&0&0&0\\
0&0&0&0&0&0&0&0&0&0&0&0&0&0&0&0&0&0&0&0&0&0&1&0&0&0&0&0&0&0\\
0&0&0&0&0&0&0&0&0&0&0&0&0&0&0&0&0&0&0&0&0&0&0&0&0&0&0&0&0&0\\
\end{array}}
\right]$$

$$M_1 = 
\left[ 
{\scriptsize \begin{array}{cccccccccccccccccccccccccccccc}
0&0&1&1&0&0&0&0&0&0&0&0&0&0&0&0&0&0&0&0&0&0&0&0&0&0&0&0&0&0\\
0&0&0&0&0&1&1&0&0&0&0&0&0&0&0&0&0&0&0&0&0&0&0&0&0&0&0&0&0&0\\
0&0&0&0&0&0&0&0&0&1&0&0&0&0&0&0&0&0&0&0&0&0&0&0&0&0&0&0&0&0\\
0&0&0&0&0&0&0&0&0&0&0&0&1&1&0&0&0&0&0&0&0&0&0&0&0&0&0&0&0&0\\
0&0&0&0&0&0&0&0&0&0&0&0&0&0&1&1&0&0&0&0&0&0&0&0&0&0&0&0&0&0\\
0&0&0&0&0&0&0&0&0&0&0&0&0&0&0&0&0&0&1&1&0&0&0&0&0&0&0&0&0&0\\
0&0&0&0&0&0&0&0&0&0&0&0&0&0&0&0&0&0&0&1&0&0&0&0&0&0&0&0&0&0\\
0&0&0&0&0&0&0&0&0&0&0&0&0&0&0&0&0&0&0&1&0&0&1&0&0&0&0&0&0&0\\
0&0&0&0&0&0&0&0&0&0&0&0&0&0&0&0&0&0&0&0&0&0&1&0&0&0&0&0&0&0\\
0&0&0&0&0&0&0&0&0&0&0&0&0&0&1&1&0&0&0&0&0&0&0&0&0&0&0&0&0&0\\
0&0&0&0&0&0&0&0&0&0&0&0&0&0&0&1&0&0&0&0&0&0&0&0&0&0&1&0&0&0\\
0&0&0&0&0&0&0&0&0&0&0&0&0&0&0&0&0&0&0&0&0&0&0&1&0&0&0&0&0&0\\
0&0&0&0&0&0&0&0&0&0&0&0&0&0&0&0&0&0&1&0&0&0&0&0&0&0&0&0&0&0\\
0&0&0&0&0&0&0&0&0&0&0&0&0&0&1&1&0&0&0&0&0&0&0&0&0&0&0&0&0&0\\
0&0&0&0&1&0&0&0&0&0&0&0&0&0&0&0&0&0&0&0&0&0&0&0&0&0&0&0&0&0\\
0&0&0&0&1&0&0&0&0&0&0&0&0&0&0&0&0&0&0&0&0&0&0&1&0&0&0&0&0&0\\
0&0&0&0&0&0&0&0&0&0&0&0&0&0&0&0&0&0&0&1&0&0&1&0&0&0&0&0&0&0\\
0&0&0&0&0&0&0&0&0&0&0&0&0&0&0&0&0&0&0&0&0&0&0&0&0&0&0&1&0&0\\
0&0&0&0&0&0&0&0&0&0&0&0&0&0&0&0&0&0&0&0&0&0&0&0&0&0&0&0&0&0\\
0&0&0&0&0&0&0&0&0&0&0&1&0&0&0&0&0&0&0&0&0&0&0&0&0&0&0&0&0&0\\
0&0&0&0&0&0&0&0&0&0&0&1&0&0&0&0&0&0&0&0&0&0&0&0&0&0&0&0&0&0\\
0&0&0&0&0&0&0&0&0&0&0&0&0&0&0&0&0&0&1&0&0&0&0&0&0&0&0&0&0&0\\
0&0&0&0&0&0&0&0&0&0&0&0&0&0&0&0&0&0&0&0&0&0&0&0&0&0&0&0&0&0\\
0&0&0&0&0&0&0&0&0&0&0&0&0&0&0&0&0&0&1&0&0&0&0&0&0&0&0&0&0&0\\
0&0&0&0&0&0&0&0&0&0&0&0&0&0&0&0&0&0&0&0&0&0&0&0&0&0&0&0&0&0\\
0&0&0&0&0&0&0&0&0&0&0&0&0&0&1&1&0&0&0&0&0&0&0&0&0&0&0&0&0&0\\
0&0&0&0&0&0&0&0&0&1&0&0&0&0&0&0&0&0&0&0&0&0&0&0&0&0&0&0&0&0\\
0&0&0&0&0&0&0&0&0&0&0&0&0&0&0&0&0&0&0&0&0&0&0&1&0&0&0&0&0&0\\
0&0&0&0&0&0&0&0&0&0&0&0&0&0&0&0&0&0&0&0&0&0&0&0&0&0&0&0&0&0\\
0&0&0&0&0&0&0&0&0&0&0&0&0&0&0&0&0&0&0&0&0&0&0&0&0&0&0&0&0&0\\
\end{array}}
\right]$$

$$v = 
\left[ 
{\scriptsize
\begin{array}{cccccccccccccccccccccccccccccc}
1&1&0&0&1&0&0&0&0&0&0&0&0&0&0&0&0&0&0&0&0&0&0&0&0&0&0&0&0&0
\end{array}
}\right]$$

$$
w =
\left[
{\scriptsize
\begin{array}{cccccccccccccccccccccccccccccc}
1&0&1&1&0&0&0&0&1&0&0&0&1&1&0&0&0&1&0&0&0&1&0&0&0&1&0&0&1&1
\end{array} 
}\right]$$

We can then obtain relations for the sequence $(a(n))_{\geq 0}$:

\begin{align*}
a(4n+3) &= a(4n+1) \\
a(8n+1) &=  a(4n+1) \\
a(8n+5) &= 0 \\
a(16n+6) &= a(4n+1) + a(4n+2) - {1\over 2} a(16n+2) + {1 \over 2} a(16n+4) \\
a(16n+8) &= 3 a(4n+1) + 3 a(4n+2) - {1 \over 2}  a(16n+2) - {3 \over 2} a(16n+4) \\
a(16n+10) &= 3 a(4n+1) + 3 a(4n+2) - {1 \over 2} a(16n+2) - {3 \over 2}  a(16n+4) \\
a(16n+12) &= a(4n+1) + a(4n+2) - {1 \over 2}  a(16n+2) + {1 \over 2} a(16n+4) \\
\end{align*}
\begin{align*}
a(32n) &= a(2n+1) - {1 \over 2}  a(4n+1) + 3 a(8n+2) - 3  a(8n+4) \\
a(32n+2) &= -ra(2n+1) +  a(4n+1) + 3 a(8n+2) -2 a(8n+4) \\
a(32n+4) &= -a(2n+1) + a(4n+1) + a(8n+2) \\
a(32n+14) &= -a(2n+1) + a(8n+4) \\
a(32n+16) &= - a(2n+1) + a(8n+4) \\
a(32n+20) &= a(32n+18) \\
a(32n+30) &= 2 a(2n+1) + a(8n+2) - 3 a(8n+4) + 2 a(8n+6) - a(32n+18) \\
a(64n+18) &= a(4n+1) \\
a(64n+50) &=  0
\end{align*}

We can also do the same thing for the number of  privileged palindromes 
$(b(n))_{n \geq 0}$
in the Thue-Morse sequence.  Here are the first few values:
\begin{center}
\begin{tabular}{c|ccccccccccccccccc}
$n$ & 0 & 1 & 2 & 3 & 4 & 5 & 6 & 7 & 8 & 9 & 10 & 11 & 12 & 13 & 14 & 15 & 16 \\
\hline
$b(n)$ & 1 & 2 & 2 & 2 & 2 & 0 & 4 & 0 & 4 & 0 & 4 & 0 & 4 & 0 & 0 & 0 & 0
\end{tabular}
\end{center}
We omit the details and just present the 
computed relations:

\begin{align*}
b(4n+3) &= b(4n+1) \\
b(8n+1) &= b(4n+1) \\
b(8n+4) &= b(8n+2) \\
b(8n+5) &= 0 \\
b(16n+6) &= b(4n+1) +  b(4n+2) \\
b(16n+8) &= b(4n+1) + b(4n+2) \\
b(16n+10) &= b(4n+1) + b(4n+2) \\
b(16n+14) &= -b(4n+1) + b(16n+2) \\
b(32n) &= b(2n+1) - {1\over 2 } b(4n+1) \\
b(32n+2) &= -b(2n+1) + b(4n+1) + b(8n+2) \\
b(32n+16) &= - b(2n+1) +  b(8n+2) \\
b(64n+18) &=  b(4n+1) \\
b(64n+50) &= 0
\end{align*}

\section{Trapezoidal words}

Trapezoidal words have many different characterizations.  The
characterization that proves useful to us is the following
\cite[Prop.~2.8]{Bucci&DeLuca&Fici:2013}:
a word $w$ is trapezoidal iff $|w| = R_w + K_w$.
Here $R_w$ is the minimal length $\ell$ for which $w$ 
contains no right-special factor of length $\ell$, and
$K_w$ is the minimal length $\ell$ for which there is a 
length-$\ell$ suffix of $w$ that appears nowhere else in $w$.

This can be translated into $\Th(\Enn, +, n \rightarrow {\bf x}[n])$
as follows:
$\rtsp(j,n,p)$ is true iff ${\bf x}[j..j+n-1]$ has a right
special factor of length $p$, and false otherwise:
\begin{multline*}
 \rtsp(j,n,p) := \exists r\ \exists s\ (\subs(r,j,p+1,n) \ \wedge \  
\subs(s,j,p+1,n) \ \wedge\ \\
\factoreq(r,s,p) \ \wedge\  {\bf x}[s+p] \not=
{\bf x}[r+p]).
\end{multline*}

$\minr(j,n,p)$ is true iff $p$ is the smallest integer
such that ${\bf x}[j..j+n-1]$ has no right special
factor of length $p$:
$$ \minr(j,n,p) := (\neg \rtsp(j,n,p)) \ \wedge\ (\forall c \ (\neg \rtsp(j,n,c)) 
	\implies (c \geq p)) .$$

$\unrepsuf(j,n,q)$ is true iff the suffix of length
$q$ of ${\bf x}[j..j+n-1]$ is unrepeated in 
${\bf x}[j..j+n-1]$:
$$ \unrepsuf(j,n,q) :=
\neg \occurs(j+n-q,j,q,n-1).
$$

$\minunrepsuf(j,n,p)$ is true iff $p$ is the length of
the shortest unrepeated suffix of
${\bf x}[j..j+n-1]$:
$$ \minunrepsuf(j,n,p) := \unrepsuf(j,n,q) \ \wedge\ 
(\forall c\ \unrepsuf(j,n,c) \implies (c\geq q)) .$$

$\trap(j,n)$ is true iff ${\bf x}[j..j+n-1]$ is trapezoidal:
$$
\trap(j,n) := \exists p \ \exists q\ (n=p+q) 
\ \wedge \ \minunrepsuf(j,n,p) \ \wedge\ \minr(j,n,q) .$$

Finally, we can determine those $n$ for which
$\bf x$ has a trapezoidal factor of length $n$ as follows:
$$ \exists j \ \trap(j,n).$$

\begin{theorem}
\begin{enumerate}[(a)]
\item There are exactly 43 trapezoidal factors of the Thue-Morse
sequence.  The longest is of length $8$.

\item There are exactly 185 trapezoidal factors of the Rudin-Shapiro
sequence.  The longest is of length $12$.

\item There are exactly 57 trapezoidal factors of the ordinary 
paperfolding sequence.  The longest is of length $8$.

\item There are exactly 77 trapezoidal factors of the period-doubling
sequence.  The longest is of length $15$.

\item Every factor of the Fibonacci word is trapezoidal.

\end{enumerate}
\end{theorem}

For parts (b) and (c) above, we used the least-significant-digit first
representation in order to have the computation terminate.

\section{Balanced words}

Our definition of balanced word above does not obviously lend itself to a
definition in first-order arithmetic.  However, for binary words, there
is an alternative characterization (due to Coven and Hedlund
\cite{Coven&Hedlund:1973}) that we can use:  a binary word $w$ is
unbalanced if and only if there exists a palindrome $v$ such that both
$0v0$ and $1v1$ are factors of $w$.

Thus we can write define $\unbal(i,n)$, a predicate which is true iff
${\bf x}[i..i+n-1]$ is unbalanced, as follows:
\begin{multline*}
\exists m \ (m \geq 2) \ \wedge\  (\exists j \ \exists k \ 
(\subs(j,i,m,n) \ \wedge\  \subs(k,i,m,n) \ \wedge \ \pal(j,m) \\
\ \wedge \ \pal(k,m) \ \wedge\ \factoreq(j+1,k+1,m-2) \ \wedge \ 
{\bf x}[j] \not= {\bf x}[k])) 
\end{multline*}

\begin{theorem}
\begin{enumerate}[(a)]
\item The Thue-Morse word has exactly 41 balanced factors.  The longest
is of length 8.  The Thue-Morse word has unbalanced factors of length
$n$ exactly when $n \geq 4$.

\item The Rudin-Shapiro word has exactly 157 balanced factors.  The longest
is of length 12.  The Rudin-Shapiro word has unbalanced factors of
length $n$ exactly when $n \geq 4$.

\item The ordinary paperfolding word has exactly 51 balanced factors.
The longest is of length 8.  The ordinary paperfolding word has unbalanced
factors of length $n$ exactly when $n \geq 4$.

\item The period-doubling word has exactly 69 balanced factors.
The longest is of length 15.  The period-doubling word has unbalanced
factors of length $n$ exactly when $n \geq 6$.

\item All factors of the Fibonacci word are balanced.
\end{enumerate}
\end{theorem}

Of course, (e) was already well known.

\section{Consequences}

As a consequence we get

\begin{theorem}
Suppose ${\bf x}$ is a $k$-automatic sequence.
Then
\begin{enumerate}[(a)]
\item The characteristic sequence of those $n$ for which
${\bf x}$ has a closed (resp., palindromic, maximal
palindromic, privileged,
rich, trapezoidal, balanced) factor
of length $n$ is $k$-automatic.

\item The sequence counting the number of closed (resp., 
palindromic, maximal palindromic,
privileged, rich, trapezoidal, balanced) factors of length $n$ is $k$-regular.

\item It is decidable, given a $k$-automatic sequence, whether
it contains arbitrarily long closed (resp., palindromic, maximal
palindromic, privileged, rich,
trapezoidal, balanced) factors.

\item There exists a function $g(k,\ell,n)$
such that if a $k$-automatic sequence $\bf w$
taking values over an alphabet of size $\ell$, generated by an
$n$-state automaton, has at least one closed
(resp., palindromic, maximal palindromic, privileged,
rich, trapezoidal, balanced) factor, then it has a factor of
length $\leq g(k,\ell,n)$.  The function $g$ does not depend on $\bf w$.

\item There exists a function $h(k,\ell,n)$
such that if a $k$-automatic sequence $\bf w$
taking values over an alphabet of size $\ell$, generated by an
$n$-state automaton, has a closed
(resp., palindromic, maximal palindromic, privileged,
rich, trapezoidal, balanced) factor of length $\geq h(k,\ell,n)$,
then it has arbitrarily large such factors.
The function $h$ does not depend on $\bf w$.
\end{enumerate}
\end{theorem}

\begin{proof}
Parts (a) and (c) follow from, for example, \cite[Theorem 1]{Shallit:2013}.
For part (b) see \cite{Charlier&Rampersad&Shallit:2012}.
Parts (d) and (e) follows from the construction converting the logical
predicate for the property to an automaton.
\end{proof}

%
%
%

\end{document}